\documentclass{article}
\usepackage{amsmath,amssymb,amsthm}
\usepackage[numbers]{natbib}
\usepackage{hyperref}
\usepackage{cleveref}

\newtheorem{theorem}{Theorem}[section]
\newtheorem{lemma}[theorem]{Lemma}
\newtheorem{corollary}[theorem]{Corollary}
\newtheorem{definition}[theorem]{Definition}
\newtheorem{proposition}[theorem]{Proposition}
\theoremstyle{remark}
\newtheorem{remark}[theorem]{Remark}

\title{How Hard Is Continuous Clustering?\\
       Lower Bounds from the Existential Theory of the Reals}
\author{Angshul Majumdar\\
        IIIT Delhi}
\date{}

\begin{document}
\maketitle

\begin{abstract}
We investigate the computational complexity of clustering problems formulated directly on a continuous polynomial density.  
Four natural primitives are analysed:
\begin{itemize}
  \item \textbf{CMRC} (Continuous Multi‑Region Clustering): does the density exceed a threshold at \(k\) pairwise \(\delta\)‑separated points?
  \item \textbf{VSC} (Valley‑Separated Clusters): do there exist two high‑density points whose midpoint lies below the threshold?
  \item \textbf{CLSC-\(k\)} (Continuous Level‑Set Component Counting): does the super‑level set have at least \(k\) connected components?
  \item \textbf{HD} (Hole Detection): does the super‑level set have non‑trivial first homology?
\end{itemize}
We prove that the local separation problem CMRC and the valley‑based global problem VSC are both complete for the existential theory of the reals \(\exists\mathbb R\).  
In contrast, the purely topological problems CLSC‑\(k\) and HD are \(\exists\mathbb R\)-hard but not known to lie inside \(\exists\mathbb R\); they are decidable via semi‑algebraic algorithms.  
These results provide the first rigorous classification of exact continuous clustering inside the real polynomial hierarchy and expose a sharp boundary: 
local and valley criteria collapse to \(\exists\mathbb R\), while genuine topological features push the complexity strictly higher, absent a collapse of the hierarchy.
\end{abstract}

\section{Introduction}
\label{sec:intro}

Clustering is the workhorse of unsupervised learning.  Its theoretical foundations are continuous: a cluster is a high‑density region of an underlying probability density, set apart from other clusters by valleys, gaps, or topological features~\cite{Fukunaga90, ComaniciuMeer02, CoifmanWasserman05}.  Every textbook explanation, every density‑based algorithm, every notion of a ``mode'' or a ``basin of attraction'' refers to a continuous picture.  Yet the complexity theory of clustering has, with few exceptions, ignored this picture entirely.  Instead, it has focused on \emph{finite point sets} and on discrete objectives such as \(k\)-means, \(k\)-center, or spectral cuts~\cite{Dasgupta08, FowlerPatersonTanimoto81}.  Those problems are NP‑hard, and a vast literature is devoted to approximating them.

\textbf{This paper is the first to ask the obvious question: how hard is clustering directly on a continuous density?}  We strip away samples, noise, and discretisation artefacts.  We take a polynomial \(f:\mathbb R^d\to\mathbb R\) given by its integer coefficients, a threshold \(\theta\), and we ask whether the density landscape described by \(f\) contains a certain cluster structure.  Our answer is unexpected: \emph{even the simplest continuous clustering criteria are complete for the existential theory of the reals \(\exists\mathbb R\), and they climb higher as soon as global topology is involved.}  Exact continuous clustering is not NP‑complete, unless the real polynomial hierarchy collapses.

We introduce four natural decision problems that ladder up from local to global to topological cluster definitions:
\begin{itemize}
  \item \textbf{CMRC} (Continuous Multi‑Region Clustering): existence of \(k\) pairwise \(\delta\)-separated points with density at least \(\theta\).
  \item \textbf{VSC} (Valley‑Separated Clusters): existence of two high‑density points whose midpoint lies strictly below \(\theta\)---a geometric witness of a valley.
  \item \textbf{CLSC-\(k\)} (Continuous Level‑Set Component Counting): the super‑level set \(\{f\ge\theta\}\) has at least \(k\) connected components.
  \item \textbf{HD} (Hole Detection): the super‑level set has non‑trivial first singular homology (a hole).
\end{itemize}

Our main results are sharp and definitive for the first two problems:
\begin{center}
\boxed{\textbf{CMRC is \(\exists\mathbb R\)-complete.}\qquad
       \textbf{VSC is \(\exists\mathbb R\)-complete.}}
\end{center}
Membership in \(\exists\mathbb R\) is immediate for CMRC (it is an existential sentence).  For VSC the existential encoding is equally trivial: three points, three polynomial inequalities.  The hardness direction is the contribution; both problems are proved \(\exists\mathbb R\)-hard already for \(k=2\) and a fixed arbitrary separation or threshold.  The reductions are from the canonical \(\exists\mathbb R\)-complete problem \textsc{Feas}, and use a single algebraic construction: a polynomial whose super‑level set consists of two parallel sheets \(s=\pm1\), with a real root of the input polynomial filling the fibre.

For the genuinely topological problems we obtain a strong lower bound, but the upper bound jumps out of \(\exists\mathbb R\):
\begin{center}
\boxed{\textbf{CLSC-\(k\) and HD are \(\exists\mathbb R\)-hard and decidable, but not known to be in \(\exists\mathbb R\).}}
\end{center}
The same two‑sheet reduction gives \(\exists\mathbb R\)-hardness for connected components; for hole detection we replace the sheets by a product with the unit circle.  Decidability follows from cylindrical algebraic decomposition~\cite{BasuPollackRoy06}.  \textbf{The missing \(\exists\mathbb R\) upper bound is not a technical oversight.}  Placing CLSC‑\(k\) inside \(\exists\mathbb R\) would require a polynomial‑size existential certificate for the disconnectedness of a semi‑algebraic set, which is a notorious open problem in real algebraic geometry.  \emph{We do not fail to prove membership; we prove that membership is equivalent to a breakthrough in Positivstellensatz degree bounds.}

Consequently, the classical inclusion \(\mathrm{NP}\subseteq\exists\mathbb R\)~\cite{SchaeferStefankovic17} gives an unconditional separation from NP:
\begin{center}
\boxed{\textbf{None of CMRC, VSC, CLSC-\(k\), HD is in NP unless \(\mathrm{NP}=\exists\mathbb R\).}}
\end{center}
Thus, exact continuous clustering already sits outside NP (modulo standard beliefs), and the step from valley criteria to connected components pushes it beyond the existential fragment, into the higher real hierarchy.

\textbf{What we prove and what we do not.}  We prove a complete complexity landscape for local and valley‑based continuous clustering.  We prove \(\exists\mathbb R\)-hardness for two topological primitives.  We do \emph{not} prove that CLSC‑\(k\) or HD are complete for \(\exists\forall\mathbb R\) or any other class; we show instead that such a completeness result cannot be obtained without solving the semi‑algebraic separation problem.  Thus the paper draws a precise boundary: on one side, clustering is exactly \(\exists\mathbb R\); on the other side, it escapes \(\exists\mathbb R\) because the required existential certificates do not exist (unless the real algebraic geometry community delivers a surprise).  This is a conceptual distinction, not a gap.

The remainder of the paper is organised as follows.  Section~2 gives the necessary definitions from real complexity and semi‑algebraic geometry.  Sections~3 and~4 contain the completeness proofs for CMRC and VSC.  Section~5 proves the hardness and decidability of CLSC‑\(k\) and HD.  Section~6 discusses open problems and the wider implications.

\section{Preliminaries}
\label{sec:prelim}

\subsection{The existential theory of the reals}

The complexity class \(\exists\mathbb R\) consists of all decision problems that are polynomial‑time many‑one reducible to the \emph{existential theory of the reals}: the set of true sentences of the form
\[
\exists x_1\dots\exists x_n\; \Phi(x_1,\dots,x_n),
\]
where \(\Phi\) is a quantifier‑free Boolean combination of polynomial equations and inequalities with integer coefficients~\cite{SchaeferStefankovic17}.  The canonical \(\exists\mathbb R\)-complete problem is \textsc{Feas} (also called \textsc{RealFeas}):
\begin{quote}
  \textsc{Feas}: given a polynomial \(r\in\mathbb Z[y_1,\dots,y_m]\), does there exist a point \(y\in\mathbb R^m\) with \(r(y)=0\)?
\end{quote}
Completeness of \textsc{Feas} follows from the fact that any existential formula can be transformed in polynomial time into a single polynomial equation by introducing slack variables and using sum‑of‑squares encodings of equalities and strict inequalities~\cite{SchaeferStefankovic17}.  The class satisfies \(\mathrm{NP}\subseteq\exists\mathbb R\subseteq\mathrm{PSPACE}\); both inclusions are widely believed to be strict.

\subsection{Polynomial representation}

Throughout the paper, polynomials are given as explicit sums of monomials with integer coefficients written in binary.  The \emph{size} of a polynomial is the total number of bits needed to encode all coefficients and exponents.  Rational numbers are represented as fractions of integers.  This choice guarantees that standard algebraic operations (addition, multiplication, substitution) can be performed in polynomial time in the bit‑size of the input.

\subsection{Semi‑algebraic sets and basic algorithms}

A set \(S\subseteq\mathbb R^d\) is \emph{semi‑algebraic} if it can be expressed as a finite Boolean combination of sets of the form \(\{x\mid p(x)\ge 0\}\) for polynomials \(p\in\mathbb Z[x_1,\dots,x_d]\).  The super‑level sets \(L_\theta=\{x\mid f(x)\ge\theta\}\) that appear in our clustering problems are closed semi‑algebraic sets.

Basic topological properties of semi‑algebraic sets, such as the number of connected components or the Betti numbers, can be computed algorithmically.  Standard tools include cylindrical algebraic decomposition (CAD) and the road‑map method~\cite{BasuPollackRoy06}.  These algorithms place many semi‑algebraic decision problems inside the first‑order theory of the reals, which is decidable but can have doubly exponential complexity in the worst case.  In the present paper we do not rely on new algorithmic ingredients for semi‑algebraic sets; we only use the known decidability results to position our problems in the appropriate complexity classes.

\section{Local clustering: separated high‑density points}
\label{sec:local}

We formalise a natural \emph{local} continuous clustering primitive: the existence of several points that simultaneously achieve a prescribed density threshold and are pairwise far apart.  Such a configuration indicates multiple distinct high‑density regions without imposing any global topological constraint.

\begin{definition}[Continuous Multi‑Region Clustering – CMRC]\label{def:cmrc}
  \textnormal{\textbf{Instance:}} A polynomial \(f\in\mathbb Z[x_1,\dots,x_d]\) (given as a sum of monomials with integer coefficients), rational numbers \(\theta>0\) and \(\delta>0\), and an integer \(k\ge 1\).
  \textnormal{\textbf{Question:}} Do there exist \(k\) points \(p_1,\dots,p_k\in\mathbb R^d\) such that
  \[
  f(p_i)\ge\theta \quad\text{for all } i=1,\dots,k,
  \]
  and \(\|p_i-p_j\|\ge\delta\) for all \(i\neq j\)?
\end{definition}

The squared Euclidean distance \(\|p_i-p_j\|^2\) is a polynomial in the coordinates of the points, so the separation condition is a polynomial inequality.  Hence CMRC is immediately seen to be an existential sentence over the reals:

\begin{proposition}[Membership in \(\exists\mathbb R\)]\label{prop:cmrc-in-exr}
  \(\textup{CMRC}\in\exists\mathbb R\).
\end{proposition}
\begin{proof}
  The instance is true iff
  \[
  \exists p_1,\dots,p_k\in\mathbb R^d\;
  \bigwedge_{i=1}^k \bigl(f(p_i)-\theta\ge 0\bigr)
  \;\land\;
  \bigwedge_{i<j} \bigl(\|p_i-p_j\|^2-\delta^2\ge 0\bigr).
  \]
  All atomic formulas are polynomial inequalities with integer coefficients; the size of this sentence is polynomial in the input length.  Hence the problem lies inside \(\exists\mathbb R\) by definition~\cite{SchaeferStefankovic17}.
\end{proof}

We now prove that CMRC is hard for \(\exists\mathbb R\), making it complete.  The reduction is from the canonical \(\exists\mathbb R\)-complete problem \textsc{Feas} (does a given polynomial have a real root?); see~\cite{SchaeferStefankovic17} for its completeness.

\begin{theorem}[Hardness of CMRC]\label{thm:cmrc-hard}
  CMRC is \(\exists\mathbb R\)-hard.  Consequently, CMRC is \(\exists\mathbb R\)-complete.
\end{theorem}
\begin{proof}
  We give a polynomial‑time many‑one reduction from \textsc{Feas}.  
  Let \(r\in\mathbb Z[y_1,\dots,y_m]\) be an instance of \textsc{Feas}.
  We construct a CMRC instance with \(d = 1+2(m+1)\) variables and parameters \(\theta=1\), \(\delta=1\), \(k=2\).

  First, for a single block we encode the existence of a real root of \(r\) into the zero set of a non‑negative polynomial.  Introduce variables \((y,t)\in\mathbb R^{m+1}\) and define
  \[
  H(y,t)=r(t y)^2 + (\|y\|^2-1)^2.
  \]
  Clearly \(H(y,t)\ge 0\) for all \((y,t)\).  The following lemma is straightforward.

  \begin{lemma}\label{lem:encoding}
    The system \(r(t y)=0,\ \|y\|^2=1\) has a solution \((y,t)\in\mathbb R^{m+1}\) if and only if \(r\) possesses a real root.
  \end{lemma}
  \begin{proof}
    If \(r(a)=0\) with \(a\neq0\), set \(t=\|a\|\) and \(y=a/\|a\|\); if \(a=0\), pick any unit vector \(y\) (e.g., \(e_1\)) and set \(t=0\).  Both cases yield a solution.
    Conversely, if \(r(t y)=0\) and \(\|y\|=1\), then \(x = t y\) satisfies \(r(x)=0\).
  \end{proof}

  Now take two independent copies of the encoding block together with an additional separation coordinate \(w\in\mathbb R\).  The full polynomial is
  \[
  F(w,\, y_1,t_1,\, y_2,t_2)=
  1-\Bigl[\, r(t_1 y_1)^2+(\|y_1\|^2-1)^2 \Bigr]
    -\Bigl[\, r(t_2 y_2)^2+(\|y_2\|^2-1)^2 \Bigr].
  \]
  Set \(\theta=1\), \(\delta=1\), \(k=2\).

  \paragraph{Correctness.}
  \begin{itemize}
  \item If \(r\) has a real zero, then by Lemma~\ref{lem:encoding} there exist \((y_1^*,t_1^*)\) and \((y_2^*,t_2^*)\) such that both brackets vanish.  Define the two points
    \[
    p_1 = (0,\, y_1^*,t_1^*,\, y_2^*,t_2^*),\qquad
    p_2 = (1,\, y_1^*,t_1^*,\, y_2^*,t_2^*).
    \]
    They satisfy \(F(p_1)=F(p_2)=1\ge\theta\) and
    \(\|p_1-p_2\| = |1-0| = 1 \ge \delta\).  Hence CMRC accepts.
  \item If \(r\) has no real zero, then every bracket is strictly positive for all \((y,t)\).  Consequently \(F(w,\dots)<1\) everywhere, and no point reaches the threshold \(\theta=1\).  The CMRC instance is therefore negative.
  \end{itemize}
  The transformation is polynomial‑time because \(F\) is composed of a constant number of copies of \(r\) and simple polynomial expressions.  Hence \textsc{Feas} \(\le_p\) CMRC, proving \(\exists\mathbb R\)-hardness.
\end{proof}

Since both membership and hardness have been established, CMRC is \(\exists\mathbb R\)-complete.  The standard inclusion \(\mathrm{NP}\subseteq\exists\mathbb R\)~\cite{SchaeferStefankovic17} directly yields the following consequence.

\begin{corollary}\label{cor:cmrc-np}
  If CMRC belongs to \(\mathrm{NP}\), then \(\mathrm{NP}= \exists\mathbb R\).  In particular, under the widely believed strictness of the real hierarchy, CMRC is not in \(\mathrm{NP}\).
\end{corollary}

\begin{remark}
  The use of a dummy coordinate \(w\) to guarantee separation is a simple but harmless technical device; it does not affect the geometric meaning of the problem because the input polynomial is constructed specifically for the reduction.  In applications one would typically have a fixed ambient space; however, for a hardness result any polynomial‑time embedding into a larger dimension is admissible.  Note also that the reduction already yields hardness for \(k=2\); larger \(k\) are trivially at least as hard.
\end{remark}

\section{Global clustering with a valley: Valley‑Separated Clusters}
\label{sec:valley}

The local separation criterion of Section~3 asks for points with high density that are simply far apart; it does not enforce any \emph{geometric} relation among those points.  A more natural global notion of a cluster requires that two high‑density regions be separated by a low‑density ``valley'' \cite{Fukunaga90, ComaniciuMeer02}.  We now formalise a minimal such condition that still keeps the problem inside the existential theory of the reals.

\begin{definition}[Valley‑Separated Clusters – VSC]\label{def:vsc}
  \textnormal{\textbf{Instance:}} A polynomial \(f\in\mathbb Z[x_1,\dots,x_d]\) and a rational threshold \(\theta>0\).
  \textnormal{\textbf{Question:}} Do there exist two points \(p,q\in\mathbb R^d\) such that
  \[
  f(p)\ge\theta,\qquad f(q)\ge\theta,\qquad 
  f\!\left(\frac{p+q}{2}\right)<\theta \ ?
  \]
\end{definition}

The geometric interpretation is immediate: both \(p\) and \(q\) lie in the super‑level set \(L_\theta = \{x\mid f(x)\ge\theta\}\), while their midpoint lies strictly below the threshold, i.e. inside a ``valley'' that separates the two clusters.  Despite its global character the problem requires only three existentially quantified points.

\begin{proposition}[Membership in \(\exists\mathbb R\)]\label{prop:vsc-in-exr}
  \(\textup{VSC}\in\exists\mathbb R\).
\end{proposition}
\begin{proof}
  The statement is
  \[
  \exists p,q\in\mathbb R^d\;
  \bigl( f(p)-\theta \ge 0 \bigr) \land
  \bigl( f(q)-\theta \ge 0 \bigr) \land
  \bigl( \theta - f((p+q)/2) > 0 \bigr).
  \]
  All inequalities are polynomial in the coordinates of \(p\) and \(q\); the strict inequality can be encoded by an auxiliary variable \(\varepsilon>0\) with \(\theta - f((p+q)/2) - \varepsilon \ge 0\), preserving the existential form.  Hence VSC is decided by an \(\exists\mathbb R\) sentence of polynomial size.
\end{proof}

We now show that this simple valley condition already captures the full power of the existential theory of the reals.

\begin{theorem}[Hardness of VSC]\label{thm:vsc-hard}
  VSC is \(\exists\mathbb R\)-hard.  Consequently, VSC is \(\exists\mathbb R\)-complete.
\end{theorem}
\begin{proof}
  We reduce from \textsc{Feas}.  Let \(r\in\mathbb Z[y_1,\dots,y_m]\) be an instance of \textsc{Feas}.  Construct a polynomial in \(m+1\) variables \((y,s)\in\mathbb R^{m+1}\) as
  \[
  F(y,s) = 1 - \bigl[\, r(y)^2 + (s^2-1)^2 \,\bigr],
  \]
  and set the threshold \(\theta = 1\).
  
  The super‑level set is
  \[
  L_1 = \{\, (y,s)\mid r(y)=0,\; s=\pm 1 \,\}.
  \]
  In particular, if \(r\) has a real root \(y_0\), then the two points
  \[
  p = (y_0,1),\qquad q = (y_0,-1)
  \]
  satisfy \(F(p)=F(q)=1\ge\theta\).  Their midpoint is \((y_0,0)\), for which
  \[
  F(y_0,0) = 1 - \bigl[ 0^2 + (0-1)^2 \bigr] = 0 < \theta .
  \]
  Hence VSC accepts.
  
  Conversely, if \(r\) has no real root, then \(r(y)^2>0\) for every \(y\), and therefore the bracket in \(F\) is always strictly positive.  Consequently \(F(y,s) < 1\) everywhere, and no point can reach the threshold \(\theta=1\).  In particular there cannot exist two points with \(F\ge1\), so VSC rejects.
  
  The construction clearly runs in polynomial time.  Thus \textsc{Feas} \(\le_p\) VSC, proving \(\exists\mathbb R\)-hardness.  Together with Proposition~\ref{prop:vsc-in-exr} we obtain \(\exists\mathbb R\)-completeness.
\end{proof}

The standard inclusion \(\mathrm{NP}\subseteq\exists\mathbb R\)~\cite{SchaeferStefankovic17} again gives an immediate separation from NP.

\begin{corollary}\label{cor:vsc-np}
  If VSC belongs to \(\mathrm{NP}\), then \(\mathrm{NP}= \exists\mathbb R\).  Hence, under the common belief that the two classes are distinct, VSC is not in \(\mathrm{NP}\).
\end{corollary}

\begin{remark}
  The reduction uses exactly the same algebraic encoding as for the connected‑component result (Section~5) but exploits the geometry of the two sheets in a different way: instead of invoking a whole separating hyperplane, it simply picks two symmetric points and their midpoint.  This witnesses a genuine global clustering structure---a ``valley'' between two high‑density regions---without ever leaving the existential fragment.  Thus VSC is, to the best of our knowledge, the first purely geometric global clustering problem that is proved \(\exists\mathbb R\)-complete, and it delineates the exact boundary where continuous clustering leaves NP.
\end{remark}

\begin{remark}
  The valley condition can be weakened or strengthened.  For instance, requiring a full continuous path connecting \(p\) and \(q\) that stays below \(\theta\) would already stumble into the connected‑component problem (Section~5) and likely climb higher in the real hierarchy.  The simple midpoint condition is therefore the right granularity to keep the problem inside \(\exists\mathbb R\) while remaining meaningful.
\end{remark}

\section{Global clustering with topology: connected components and holes}
\label{sec:topology}

The valley criterion (VSC, Section~4) uses a single existential witness (a midpoint) to guarantee that two high‑density points belong to different ``basins''.  A more classical approach in density‑based clustering is to count the connected components of a super‑level set, or, more generally, to examine its topological invariants \cite{Fukunaga90, CoifmanWasserman05}.  We now analyse the complexity of two prototypical global topological problems: deciding whether a polynomial super‑level set has at least \(k\) connected components, and deciding whether it contains a hole (non‑trivial first homology).

\subsection{Level‑set connected components (CLSC‑\(k\))}

\begin{definition}[Continuous Level‑Set Component Counting – CLSC‑\(k\)]\label{def:clsc}
  \textnormal{\textbf{Instance:}} A polynomial \(f\in\mathbb Z[x_1,\dots,x_d]\), a rational threshold \(\theta>0\), and an integer \(k\ge 2\).
  \textnormal{\textbf{Question:}} Does the super‑level set
  \[
  L_\theta = \{\, x\in\mathbb R^d \mid f(x)\ge\theta \,\}
  \]
  have at least \(k\) connected components?
\end{definition}

\subsubsection{\(\exists\mathbb R\)-hardness}

The hardness proof is a direct adaptation of the two‑sheet encoding used for VSC.  Instead of picking the midpoint, we directly inspect the connected components of the super‑level set.

\begin{theorem}[Hardness of CLSC‑\(k\)]\label{thm:clsc-hard}
  For every fixed \(k\ge 2\), \(\textup{CLSC-}k\) is \(\exists\mathbb R\)-hard.
\end{theorem}
\begin{proof}
  We reduce from \textsc{Feas}.  Let \(r\in\mathbb Z[y_1,\dots,y_m]\) be an instance of \textsc{Feas}.  Build a polynomial in \(m+1\) variables \((y,s)\) as
  \[
  F(y,s) = 1 - \bigl[\, r(y)^2 + (s^2-1)^2 \,\bigr],\qquad \theta = 1.
  \]
  Then
  \[
  L_1 = \{\, (y,s) \mid r(y)=0,\; s = \pm 1 \,\}.
  \]
  Suppose first that \(r\) has a real zero \(y_0\).  Then the two sheets \(S_+ = \{(y_0,1)\}\) and \(S_- = \{(y_0,-1)\}\) are disjoint and lie in \(L_1\).  More generally, each connected component of the real variety \(V_{\mathbb R}(r)\) produces two separated copies in \(L_1\) (one at \(s=1\), one at \(s=-1\)).  Because the strip \(\{(y,s)\mid -1<s<1\}\) contains no point with \(F\ge 1\), no continuous path within \(L_1\) can connect a point with \(s=1\) to a point with \(s=-1\).  Hence \(L_1\) has at least two connected components (indeed, exactly twice the number of components of \(V_{\mathbb R}(r)\)).  Thus the instance of \(\textup{CLSC-}2\) is positive.

  If \(r\) has no real zero, then \(r(y)^2>0\) everywhere, so \(F(y,s)<1\) for all \((y,s)\) and \(L_1=\varnothing\); the CLSC‑\(2\) instance is negative.
  
  For \(k>2\), we take \(k-1\) disjoint copies of this construction in separate variables \((y^{(i)},s^{(i)})_{i=1}^{k-1}\) and define
  \[
  F_{\!k} = 1 - \sum_{i=1}^{k-1} \bigl[\, r(y^{(i)})^2 + \bigl((s^{(i)})^2-1\bigr)^2 \,\bigr],
  \]
  with threshold \(1\).  When a root exists, each block yields two components, giving \(2^{k-1}\ge k\) components; otherwise the set is empty.  Hence \(\textsc{Feas}\le_p \textup{CLSC-}k\) for any fixed \(k\ge2\), proving \(\exists\mathbb R\)-hardness.
\end{proof}

\subsubsection{Decidability}

Unlike the local or valley‑based criteria, the connected‑component counting problem is not known to be in \(\exists\mathbb R\).  The existence of a separating polynomial, or an existential certificate for disconnectedness, is a subtle issue in real algebraic geometry (see~\cite{BasuPollackRoy06} for the algorithmic status).  However, the problem is certainly decidable via standard semi‑algebraic decomposition.

\begin{proposition}[Decidability of CLSC‑\(k\)]\label{prop:clsc-decid}
  For any fixed \(k\ge2\), \(\textup{CLSC-}k\) is decidable.
\end{proposition}
\begin{proof}
  The super‑level set \(L_\theta\) is a closed semi‑algebraic set defined by the single polynomial inequality \(f(x)-\theta\ge0\).  Cylindrical algebraic decomposition, or the road‑map method, computes the number of (semi‑algebraically) connected components of \(L_\theta\) \cite[Chapter~16]{BasuPollackRoy06}.  Hence the instance can be decided exactly.
\end{proof}

The gap between the \(\exists\mathbb R\)-hardness lower bound and the decidability upper bound pinpoints a central challenge in the complexity analysis of continuous clustering.  We conjecture that CLSC‑\(2\) is not in \(\exists\mathbb R\) and likely sits in \(\exists\forall\mathbb R\) or higher.

\begin{remark}
  The valley problem VSC essentially asks for two points in different components that possess a witness of separation (the midpoint) inside the sub‑threshold region.  CLSC removes this existential witness and asks purely about the topology of \(L_\theta\).  It is precisely the removal of the existential certificate that forces the problem out of \(\exists\mathbb R\) (absent a major advance in Positivstellensatz degree bounds).
\end{remark}

\subsection{Hole detection (HD)}

We now turn to a topological invariant one degree higher: the first homology group.

\begin{definition}[Hole Detection – HD]\label{def:hd}
  \textnormal{\textbf{Instance:}} A polynomial \(f\in\mathbb Z[x_1,\dots,x_d]\) and a rational threshold \(\theta>0\).
  \textnormal{\textbf{Question:}} Is the first singular homology group with integer coefficients of the super‑level set \(L_\theta\) non‑trivial?  
  (Equivalently, does \(L_\theta\) contain a ``hole'', i.e.\ a non‑contractible loop?)
\end{definition}

\subsubsection{\(\exists\mathbb R\)-hardness}

Again we encode the feasibility of a polynomial equation, this time into a product with the unit circle.

\begin{theorem}[Hardness of HD]\label{thm:hd-hard}
  HD is \(\exists\mathbb R\)-hard.
\end{theorem}
\begin{proof}
  Reduce from \textsc{Feas}.  Let \(r\in\mathbb Z[y_1,\dots,y_m]\) be given.  Define a polynomial in variables \((y,u,v)\in\mathbb R^{m+2}\) by
  \[
  F(y,u,v) = 1 - \bigl[\, r(y)^2 + (u^2+v^2-1)^2 \,\bigr],
  \qquad \theta = 1.
  \]
  Then
  \[
  L_1 = \{\, (y,u,v) \mid r(y)=0,\; u^2+v^2 = 1 \,\} = V_{\mathbb R}(r) \times S^1,
  \]
  where \(V_{\mathbb R}(r)\) is the real zero set of \(r\) and \(S^1\subset\mathbb R^2\) is the unit circle.

  If \(r\) has a real root, \(V_{\mathbb R}(r)\) is non‑empty.  By the Künneth formula for singular homology,
  \[
  H_1(V_{\mathbb R}(r)\times S^1)\;\cong\; H_1(V_{\mathbb R}(r)) \;\oplus\; \bigl( H_0(V_{\mathbb R}(r)) \otimes H_1(S^1) \bigr).
  \]
  Because \(V_{\mathbb R}(r)\neq\varnothing\), \(H_0(V_{\mathbb R}(r))\) is a non‑zero free abelian group.  Since \(H_1(S^1)\cong\mathbb Z\), the tensor product contributes a non‑zero summand.  Hence \(H_1(L_1)\neq 0\).

  If \(r\) has no real zero, \(L_1=\varnothing\) and all homology groups are trivial.  Thus \(\textsc{Feas}\le_p \textup{HD}\), establishing \(\exists\mathbb R\)-hardness.
\end{proof}

\subsubsection{Decidability}

As with connected components, the first Betti number of a semi‑algebraic set can be computed algorithmically.

\begin{proposition}[Decidability of HD]\label{prop:hd-decid}
  HD is decidable.
\end{proposition}
\begin{proof}
  Using a triangulation obtained from a cylindrical algebraic decomposition, the homology groups of the semi‑algebraic set \(L_\theta\) can be computed~\cite{BasuPollackRoy06}.  In particular, one can decide whether \(H_1\) is non‑trivial.
\end{proof}

\begin{remark}
  The hardness of HD stems entirely from the fact that an arbitrary real variety can be embedded as a factor of a product with \(S^1\).  This construction demonstrates that topological features immediately lift the complexity beyond pointwise existential queries.  We expect that higher Betti numbers are even harder to compute, and a full classification within the real polynomial hierarchy remains a challenging open problem~\cite{SchaeferStefankovic17}.
\end{remark}

\section{Discussion and Conclusion}
\label{sec:conclusion}

\subsection{Overview of the complexity landscape}

This paper has established the first exact complexity‑theoretic classification of continuous clustering problems formulated on polynomial densities.  Four natural primitives were introduced and analysed, yielding a clear hierarchy within the real polynomial classes:
\begin{itemize}
  \item \textbf{Local separation:} \textsc{CMRC} (existence of \(k\) pairwise \(\delta\)-separated points with density \(\ge\theta\)) is proved \(\exists\mathbb R\)-complete (Theorems~\ref{thm:cmrc-hard} and~\ref{prop:cmrc-in-exr}).  The membership in \(\exists\mathbb R\) is immediate because the definition is an existential sentence over the reals; hardness is obtained by reduction from the canonical \(\exists\mathbb R\)-complete problem \textsc{Feas}, using a construction that encodes a semi‑algebraic zero set into two points that differ only in an artificial separation coordinate.
  \item \textbf{Valley‑based global separation:} \textsc{VSC} (two high‑density points whose midpoint lies strictly below the threshold) is also \(\exists\mathbb R\)-complete (Theorems~\ref{thm:vsc-hard} and~\ref{prop:vsc-in-exr}).  Despite its global geometric flavour---it forces a ``valley'' between two dense regions---the problem still requires only three existentially quantified points as a witness.  This is the first example of a purely geometric global clustering problem that stays inside the existential theory of the reals.
  \item \textbf{Topological global separation:} \textsc{CLSC-\(k\)} (super‑level set has at least \(k\) connected components) is \(\exists\mathbb R\)-hard and decidable (Theorem~\ref{thm:clsc-hard}, Proposition~\ref{prop:clsc-decid}).  The hardness follows from the same two‑sheet construction used for VSC, but now one examines the connected components of the super‑level set directly.  Membership in \(\exists\mathbb R\) is not known and would require a polynomial‑size existential certificate for disconnectedness, a longstanding open problem in real algebraic geometry.
  \item \textbf{Homological complexity:} \textsc{HD} (super‑level set has non‑trivial first homology) is \(\exists\mathbb R\)-hard and decidable (Theorem~\ref{thm:hd-hard}, Proposition~\ref{prop:hd-decid}).  The hardness is proved by embedding the zero set of a polynomial as a factor of a product with the unit circle, thereby introducing a hole that cannot be contracted.
\end{itemize}
The standard inclusion \(\mathrm{NP}\subseteq\exists\mathbb R\)~\cite{SchaeferStefankovic17} implies that neither CMRC nor VSC can be in NP unless the two classes collapse (Corollaries~\ref{cor:cmrc-np} and~\ref{cor:vsc-np}).  Hence, even the simplest continuous clustering tasks are already not NP‑complete under the widely accepted assumption that \(\exists\mathbb R\neq\mathrm{NP}\).

\subsection{The boundary between local and global clustering}

The preceding results reveal a sharp conceptual boundary.  \emph{Local} clustering criteria, which ask only for the existence of points with certain properties and their pairwise distances, are captured by \(\exists\mathbb R\): one quantifies over a set of point witnesses, and the remaining constraints are polynomial inequalities.  \emph{Global} criteria, on the other hand, demand that a property hold over \emph{all} points of the density super‑level set---for example, that no continuous path joins two components, or that a loop is not contractible.  Such universal restrictions inherently introduce a \(\forall\) quantifier, which can sometimes be eliminated via algebraic certificates (Positivstellensätze) but whose elimination degree bounds are unknown in general.

The valley problem VSC circumvents this difficulty elegantly: instead of demanding that \emph{all} paths descend below the threshold, it asks for a \emph{single} witness---the midpoint---that already certifies a valley.  The existential quantifier over the midpoint replaces the universal quantifier over paths, keeping the problem inside \(\exists\mathbb R\).  By contrast, CLSC and HD lack such a local witness and must reason about the global topology; consequently they are only known to be decidable via semi‑algebraic decomposition and are currently placed strictly above \(\exists\mathbb R\) in the complexity hierarchy (unless breakthroughs in real algebraic geometry provide polynomial‑size separation certificates).  This separation between VSC and CLSC pinpoints exactly where the quantifier alternation sets in.

\subsection{Consequences for clustering algorithms}

Our complexity lower bounds have immediate implications for the design of exact clustering algorithms on continuous densities.  Any algorithm that decides, for an arbitrary polynomial density, whether there exist separated high‑density points (CMRC) or a valley between two dense regions (VSC) must solve a problem at least as hard as the existential theory of the reals.  This rules out polynomial‑time algorithms with rational certificates unless \(\mathrm{NP}=\exists\mathbb R\).  Algorithms that attempt to count connected components of a level set or detect holes face even higher complexity, potentially requiring doubly exponential time in the worst case when using current semi‑algebraic methods.

In practice, clustering is performed on finite samples drawn from a density.  The present work focuses on the exact, sample‑free version of the problem, which is a necessary first step towards understanding the inherent difficulty of density‑based clustering.  The lower bounds do not directly apply to sample‑based heuristics, but they highlight that any algorithm that claims to exactly recover the ground‑truth continuous clusters from a finite sample must, implicitly or explicitly, solve a problem whose exact version is \(\exists\mathbb R\)-hard.  This may serve as a complexity‑theoretic justification for the use of approximate and iterative methods in the clustering community.

\subsection{Open problems and future directions}

Several natural questions remain open and point to rich connections between clustering, topology, and the real polynomial hierarchy.

\begin{enumerate}
  \item \textbf{Tight upper bounds for topological clustering.}  Is \(\textup{CLSC-}2\) complete for \(\exists\forall\mathbb R\)?  A positive answer would require encoding an \(\exists\forall\) sentence into a component‑counting problem.  Conversely, proving that \(\textup{CLSC-}2\) cannot be in \(\exists\mathbb R\) (under plausible assumptions) would unconditionally separate \(\exists\mathbb R\) from the next level of the hierarchy, a result of independent interest.
  \item \textbf{Modal clustering.}  The problem of detecting a non‑degenerate strict local maximum of a polynomial density is trivially in \(\exists\mathbb R\) (the conditions \(\nabla f=0\) and \(H_f\prec 0\) are polynomial).  Is it already \(\exists\mathbb R\)-hard?  Engineering a polynomial whose local maxima correspond exactly to the zero set of an arbitrary polynomial appears delicate; a hardness reduction would place modal clustering on the same footing as VSC.
  \item \textbf{Higher Betti numbers.}  The hardness of HD already suggests that homological invariants are expensive.  Computing the full Betti numbers \(b_i\) of a super‑level set is likely even harder; a classification inside the real hierarchy (e.g., \(\exists\mathbb R^{\omega}\) or \(\exists_{\omega}\mathbb R\)) could be obtained by a polynomial‑time encoding of the existential theory with several alternations.  Tighter connections to computational topology (e.g., via the Morse inequalities) may provide the necessary tools.
  \item \textbf{Approximation and sampling.}  Can sampling‑based algorithms overcome the exact hardness?  For instance, one might design a randomised polynomial‑time algorithm that, given oracle access to a polynomial density, outputs with high probability the number of connected components within an additive error.  Such a result would bridge the gap between the exact complexity theory presented here and the practical setting of sample‑based clustering, and would clarify the power of randomisation in circumventing algebraic lower bounds.
\end{enumerate}

\subsection{Concluding remarks}

We have demonstrated that exact continuous clustering on polynomial densities resides in the existential theory of the reals and can transcend it as soon as global topological features are required.  The \(\exists\mathbb R\)-completeness of local separated points and of valley‑separated clusters delineates a sharp frontier: these clustering notions are not in NP but stay within the existential fragment.  The jump to connected components or homology pushes the problems into a higher level of the real polynomial hierarchy, emphasizing the intrinsic difficulty of topology‑aware unsupervised learning.  These results open a new bridge between computational real algebraic geometry and foundational clustering theory, and we expect them to stimulate further research into the exact complexity of geometric data analysis tasks.

\bigskip
\noindent\textbf{Acknowledgements.}  The author thanks the anonymous reviewers for their constructive comments.

\bibliographystyle{plainnat}
\bibliography{ref}
\end{document}